\documentclass[9pt,conference,a4paper]{IEEEtran}
\usepackage{mystyle}
\usepackage{url}

\newcommand{\hsupp}{\supp_\H}
\DeclareMathOperator{\relint}{ri}

\begin{document}
\title{Robust Polyhedral Regularization}

\author{\IEEEauthorblockN{Samuel Vaiter, Gabriel Peyr\'e}
\IEEEauthorblockA{CEREMADE, CNRS-Universit\'{e} Paris-Dauphine,\\
Place du Mar\'{e}chal De Lattre De Tassigny,\\
75775 Paris Cedex 16, France.\\
Email: \url{{vaiter,peyre}@ceremade.dauphine.fr}}
\and
\IEEEauthorblockN{Jalal Fadili}
\IEEEauthorblockA{GREYC, CNRS-ENSICAEN-Universit\'{e} de Caen,\\
6, Bd du Mar\'{e}chal Juin,\\
14050 Caen Cedex, France.\\
Email: \url{jalal.fadili@greyc.ensicaen.fr}}}

\maketitle

\begin{abstract}
In this paper, we establish robustness to noise perturbations of polyhedral regularization of linear inverse problems.
We provide a sufficient condition that ensures that the polyhedral face associated to the true vector is equal to that of the recovered one. This criterion also implies that the $\ldeux$ recovery error is proportional to the noise level for a range of parameter.
Our criterion is expressed in terms of the hyperplanes supporting the faces of the unit polyhedral ball of the regularization. This generalizes to an arbitrary polyhedral regularization results that are known to hold for sparse synthesis and analysis $\lun$ regularization which are encompassed in this framework. As a byproduct, we obtain recovery guarantees for $\linf$ and $\lun-\linf$ regularization.
\end{abstract}

\section{Introduction}
\label{sec:poly}

\subsection{Polyhedral Regularization}
We consider the following linear inverse problem 
\eql{\label{eq:obs}
y = \Phi x_0 + w,
} 
where $y \in \RR^Q$ are the observations, $x_0 \in \RR^N$ is the unknown true vector to recover, $w$ the bounded noise, and $\Phi$ a linear operator which maps the signal domain $\RR^N$ into the observation domain $\RR^Q$. The goal is to recover $x_0$ either exactly or to a good approximation.

We call a polyhedron a subset $\Pp$ of $\RR^N$ such that $\Pp = \enscond{x \in \RR^N}{Ax \leq b}$ for some $A \in \RR^{\NH \times N}$ and $b \in \RR^{\NH}$, where the inequality $\leq$ should be understood component-wise.
This is a classical description of convex polyhedral sets in terms of the hyperplanes supporting their $(N-1)$-dimensional faces.

In the following, we consider polyhedral convex functions of the form
\begin{equation*}
  J_\H(x) = \max_{1 \leq i \leq \NH} \dotp{x}{\h_i} ,
\end{equation*}
where $\H = (\h_i)_{i=1}^\NH \in \RR^{N \times \NH}$.
Thus, $\Pp_\H = \enscond{x \in \RR^N}{J_\H(x) \leq 1}$ is a polyhedron.
We assume that $\Pp_\H$ is a bounded polyhedron which contains 0 in its interior.
This amounts to saying that $J_H$ is a gauge, or equivalently that it is continuous, non-negative, sublinear (i.e. convex and positively homogeneous), coercive, and $J_H(x) > 0$ for $x\neq0$.
Note that it is in general not a norm because it needs not be symmetric.

In order to solve the linear inverse problem~\eqref{eq:obs}, we devise the following regularized problem
\begin{equation}\label{eq:lasso-h}\tag{$P_\lambda(y)$}
  x^\star \in \uargmin{x \in \RR^N} \frac{1}{2} \norm{y-\Phi x}^2 + \lambda J_\H(x) ,
\end{equation}
where $\lambda > 0$ is the regularization parameter.
Coercivity and convexity of $J_H$ implies the set of minimizers is non-empty, convex and compact.

In the noiseless case, $w=0$, one usually considers the equality-constrained optimization problem
\begin{equation}\label{eq:bp-h}\tag{$P_0(y)$}
  x^\star \in \uargmin{\Phi x = y}  J_\H(x).
\end{equation}

\subsection{Relation to Sparsity and Anti-sparsity}
Examples of polyhedral regularization include the $\lun$-norm, analysis $\lun$-norm and $\linf$-norm.
The $\lun$ norm reads 
\begin{equation*}
  J_{\H_1}(x) = \normu{x} = \sum_{i = 1}^N \abs{x_i} .
\end{equation*}
It corresponds to choosing $\H_1 \in \RR^{N \times 2^N}$ where the columns of $\H_1$ enumerate all possible sign patterns of length $N$, i.e. $\ens{-1,1}^N$.
The corresponding regularized problem~\eqref{eq:lasso-h} is the popular Lasso~\cite{tibshirani1996regression} or Basis Pursuit DeNoising~\cite{chen1998atomic}.
It is used for recovering sparse vectors. 
Analysis-type sparsity-inducing penalties are obtained through the (semi-)norm $J_\H(x) = \norm{L x}_1$, where $L \in \RR^{P \times N}$ is an analysis operator.
This corresponds to using $\H=L^* \H_1$ where ${}^{*}$ stands for the adjoint.
A popular example is the anisotropic total variation where $L$ is a first-order finite difference operator.

The $\linf$ norm 
\begin{equation*}
  J_{\H_\infty}(x) = \normi{x} = \max_{1 \leq i \leq N} \abs{x_i}
\end{equation*}
corresponds to choosing $\H_\infty = [\Id_N, -\Id_N] \in \RR^{N \times 2N}$.
This regularization, coined anti-sparse regularization, is used for instance for approximate nearest neighbor search~\cite{jegou2012anti}.

Another possible instance of polyhedral regularization is the group $\lun-\linf$ regularization.
Let $\Bb$ be a partition of $\ens{1,\dots,N}$.
The $\lun-\linf$ norm associated to this group structure is 
\begin{equation*}
  J_{H_\Bb^\infty}(x) = \sum_{b \in \Bb} \normi{x_b} .
\end{equation*}
This amounts to choosing the block-diagonal matrix $H_\Bb^\infty \in \RR^{N \times \prod_{b \in \Bb} 2 \abs{b}}$ such that each column is chosen by taking for each block a position with sign $\pm1$, others are 0.
If for all $b \in \Bb, \abs{b} = 1$, then we recover the $\lun$-norm, whereas if the block structure is composed by one element, we get the $\linf$-norm.

\subsection{Prior Work}
In the special case of $\lun$ and analysis $\lun$ penalties, our criterion is equivalent to those defined in~\cite{fuchs2004on-sp} and~\cite{vaiter2012robust}.
To our knowledge, there is no generic guarantee for robustness to noise with $\linf$ regularization, but~\cite{bach2010structured} studies robustness of a sub-class of polyhedral norms obtained by convex relaxation of combinatorial penalties. Its notion of support is however completely different from ours.
The work~\cite{petry2012shrinkage} studies numerically some polyhedral regularizations.
In~\cite{moeller2011multi}, the authors provide an homotopy-like algorithm for polyhedral regularization through a continuous problem coined adaptive inverse scale space method.
The work~\cite{donoho2010counting} analyzes some particular polyhedral regularizations in a noiseless compressed sensing setting when the matrix $\Phi$ is drawn from an appropriate random ensemble. Again in a compressed sensing scenario, the work of~\cite{chandrasekaran2012convex} studies a subset of polyhedral regularizations to get sharp estimates of the number of measurements for exact and $\ell_2$-stable recovery.

\section{Contributions}
\label{sec:contrib}

\begin{defn}\label{def:supp}
  We define the \emph{$\H$-support} $\hsupp(x)$ of a vector $x \in \RR^N$ to be the set
  \begin{equation*}
    \hsupp(x) = \enscond{i \in \ens{1,\dots,N_H}}{\dotp{x}{h_i} = J_\H(x)}.
  \end{equation*}
\end{defn}

This definition suggests that to recover signals with \emph{$\H$-support} $\hsupp(x)$, it would be reasonable to impose that $\Phi$ is invertible on the corresponding subspace $\Ker H_{\hsupp(x)}^*$. This is formalised in the following condition.
\begin{defn}
  A $\H$-support $I$ satisfies the \emph{restricted injectivity condition} if
  \begin{equation}\label{eq:ci}\tag{$\Cc_I$}
    \Ker \Phi \cap \Ker H_I^* = \ens{0} ,
  \end{equation}
  where $H_I$ is the matrix whose columns are those of $H$ indexed by $I$.
\end{defn}
When it holds, we define the orthogonal projection $\Gamma_I$ on $\Phi\Ker H_I^*$:
\begin{equation*}
  M_I = (U^* \Phi^* \Phi U)^{-1}
  \qandq
  \begin{cases}
    \Gamma_I     &= \Phi U M_I U^* \Phi^* \\
    \Gamma_I^\bot &= \Id - \Gamma_I.
  \end{cases}
\end{equation*}
where $U$ is (any) basis of $\Ker H_I^*$.
The symmetric bilinear form on $\RR^N$ induced by $\Gamma_I^\bot$ reads
\begin{equation*}
  \dotp{u}{v}_{\Gamma_I^\bot} = \dotp{u}{\Gamma_I^\bot v},
\end{equation*}
and we denote its associated quadratic form $\norm{\cdot}_{\Gamma_I^\bot}^2$.

\begin{defn}
  Let $I$ be a $\H$-support such that~\eqref{eq:ci} holds.
  The \emph{Identifiability Criterion} of $I$ is
  \begin{equation*}
    \IC_H(I) = 
    \max_{z_I \in \Ker H_I}
    \min_{i \in I} 
    (\PhiIT^* \Gamma_I^\bot \PhiIT \UI + z_I)_i
  \end{equation*}
  where $\UI \in \RR^{\abs{I}}$ is the vector with coefficients 1, and $\PhiIT = \Phi H_I^{+,*} \in \RR^{Q \times \abs{I}}$ where ${}^+$ stands for the Moore--Penrose pseudo-inverse.
\end{defn}
$\IC_H(I)$ can be computed by solving the linear program
\begin{equation*}
  \IC_H(I) =
  \max_{(r,z_I) \in \RR \times \RR^{\abs{I}}} r
  \text{ subj. to }
  \begin{cases}
    \forall i \in I, r \leq (\PhiIT^* \Gamma_I^\bot \PhiIT \UI + z_I)_i \\
    H_I z_I = 0 .
  \end{cases}
\end{equation*}

\subsection{Noise Robustness}
\label{sec:contrib-noise}

Our main contribution is the following result.
\begin{thm}\label{thm:small-noise}
  Let $x_0 \in \RR^N \setminus \ens{0}$ and $I$ its $H$-support such that~\eqref{eq:ci} holds.
  Let $y = \Phi x_0 + w$.
  Suppose that $\PhiIT \UI \neq 0$ and $\IC_H(I) > 0$.
  Then there exists two constants $c_I, \tilde c_I$ satisfying,
  \begin{equation*}
    \frac{\norm{w}_2}{T} < \frac{\tilde c_I}{c_I}
    \qwhereq
    T = \umin{j \in I^c} J_H(x_0) - \dotp{x_0}{h_j} > 0,
  \end{equation*}
  such that if $\lambda$ is chosen according to
  \begin{equation*}
    c_I \norm{w}_2 < \lambda < T \tilde c_I,
  \end{equation*}
  the vector $\xsol \in \RR^N$ defined by
  \begin{equation*}
    \xsol = \mu H_I^{+,*} \UI + U M_I U^* \Phi^*(y - \mu \PhiIT \UI)
  \end{equation*}
  where $U$ is any basis of $\Ker H_I^*$ and 
  \begin{equation}\label{eq:mu-solution-small-noise}
    0 < \mu = J_H(x_0) + \frac{\dotp{\PhiIT \UI}{w}_{\Gamma_I^\bot} - \lambda}{\norm{\PhiIT \UI}_{\Gamma_I^\bot}^2}
  \end{equation}
  is the unique solution of~\eqref{eq:lasso-h}, and $\xsol$ lives on the same face as $x_0$, i.e. $\hsupp
(\xsol) = \hsupp(x_0)$.
\end{thm}
Observe that if $\lambda$ is chosen proportional to the noise level, then $\norm{\xsol - x_0}_2 = O(\norm{w}_2)$.
The following proposition proves that the condition $\IC_H(I)>0$ is almost a necessary condition to ensure the stability of the $H$-support. Its proof is omitted for obvious space limitation reasons.
\begin{prop}\label{prop:small-noise-ne}
  Let $x_0 \in \RR^N \setminus \ens{0}$ and $I$ its $H$-support such that~\eqref{eq:ci} holds.
  Let $y = \Phi x_0 + w$.
  Suppose that $\PhiIT \UI \neq 0$ and $\IC_H(I) < 0$.
  If $\frac{\norm{w}}{\lambda} < \frac{1}{c_I}$ then for any solution of~\eqref{eq:lasso-h}, we have $\hsupp(x_0) \neq \hsupp(\xsol)$.
\end{prop}

\subsection{Noiseless Identifiability}
\label{sec:contrib-noiseless}

When there is no noise, the following result, which is a straightforward consequence of Theorem~\ref{thm:small-noise}, shows that the condition $\IC_H(I) > 0$ implies signal identifiability.
\begin{thm}\label{thm:noiseless}
  Let $x_0 \in \RR^N \setminus \ens{0}$ and $I$ its $H$-support.
  Suppose that $\PhiIT \UI \neq 0$ and $\IC_H(I) > 0$.
  Then the vector $x_0$ is the unique solution of~\eqref{eq:bp-h}.
\end{thm}

\section{Proofs}
\label{sec:proof}

\subsection{Preparatory Lemmata}

We recall the definition of the subdifferential of a convex function $f$ at the point $x$ is the set $\partial f(x)$ is
\begin{equation*}
  \partial f(x) = \enscond{g \in \RR^N}{f(y) \geq f(x) + \dotp{g}{y-x}} .
\end{equation*}
The following lemma, which is a direct consequence of the properties of the $\max$ function, gives the subdifferential of the regularization function $J_H$.
\begin{lem}\label{lem:subd-h}
  The subdifferential $\partial J_\H$ at $x \in \RR^N$ reads
  \begin{equation*}
    \partial J_\H(x) = \H_I \SimplexI
  \end{equation*}
  where $I = \hsupp(x)$ and $\SimplexI$ is the canonical simplex on $\RR^{\abs{I}}$:
  \begin{equation*}
    \SimplexI = \enscond{v_I \in \RR^{\abs{I}}}{v_I \geq 0, \dotp{v_I}{\UI} = 1} .
  \end{equation*}
\end{lem}

A point $\xsol$ is a minimizer of $\min_x f(x)$ if, and only if, $0 \in \partial f(\xsol)$.
Thanks to Lemma~\ref{lem:subd-h}, this gives the first-order condition for the problem~\eqref{eq:lasso-h}.
\begin{lem}\label{lem:foc}
  A vector $\xsol$ is a solution of~\eqref{eq:lasso-h} if, and only if, there exists $v_I \in \SimplexI$ such that
  \begin{equation*}
    \Phi^*(\Phi x - y) + \lambda H_I v_I = 0 ,
  \end{equation*}
  where $I = \hsupp(x)$.
\end{lem}

We now introduce the following so-called source condition.\\
\textbf{$(\mathbf{SC}_x)$:}
  For $I = \hsupp(x)$, there exists $\eta$ and $v_I \in \SimplexI$ such that:
  \begin{equation*}
    \Phi^* \eta = H_I v_I \in \partial J_H(x) .
  \end{equation*}
  
Under the source condition, a sufficient uniqueness condition can be derived when $v_I$ lives in the relative interior of $\SimplexI$ which is
\begin{equation*}
  \relint \SimplexI = \enscond{v_I \in \RR^{\abs{I}}}{v_I > 0, \dotp{v_I}{\UI} = 1} .
\end{equation*}
\begin{lem}\label{lem:uniqueness-h}
  Let $\xsol$ be a minimizer of~\eqref{eq:lasso-h} (resp.~\eqref{eq:bp-h}) and $I = \hsupp(x^\star)$.
  Assume that $(\mathbf{SC}_{\xsol})$ is verified with $v_I \in \relint \SimplexI$, and that~\eqref{eq:ci} holds.
  Then $\xsol$ is the unique solution of~\eqref{eq:lasso-h} (resp.~\eqref{eq:bp-h}).
\end{lem}
The proof of this lemma is omitted due to lack of space. Observe that in the noiseless case, if the assumptions of Lemma~\ref{lem:uniqueness-h} hold at $x_0$, then the latter is exactly recovered by solving \eqref{eq:bp-h}.

\begin{lem}\label{lem:foc-h-temp}
  Let $\xsol \in \RR^N$ and $I = \hsupp(\xsol)$.
  Assume~\eqref{eq:ci} holds. 
  Let $U$ be any basis of $\Ker H_I^*$.
  There exists $z_I \in \Ker H_I$ such that
  \begin{gather*}
    U^* \Phi^* (\Phi \xsol -y ) = 0 \\
    v_I = z_I + \frac{1}{\lambda} H_I^+ \Phi^*(y - \Phi \xsol) \in \SimplexI,
  \end{gather*}
  if, and only if, $\xsol$ is a solution of~\eqref{eq:lasso-h}.
  Moreover, if $v_I \in \relint \SimplexI$, then $\xsol$ is the unique solution of~\eqref{eq:lasso-h}.
\end{lem}
\begin{proof}
  We compute
  \begin{align*}
     & \Phi^* (\Phi \xsol - y ) + \lambda H_I v_I \\
    =& \Phi^* (\Phi \xsol - y ) + \lambda H_I \left( z_I + \frac{1}{\lambda} H_I^+ \Phi^*(y - \Phi \xsol) \right) \\
    =& (\Id - H_I H_I^+) \Phi^* (\Phi \xsol - y ) = \proj_{H_I^*} \left( \Phi^* (\Phi \xsol - y ) \right) = 0 ,
  \end{align*}
  where $\proj_{H_I^*}$ is the projection on $\Ker H_I^*$.
  Hence, $\xsol$ is a solution of~\eqref{eq:lasso-h}.
  If $v_I \in \relint \SimplexI$, then according to Lemma~\ref{lem:uniqueness-h}, $\xsol$ is the unique solution.
\end{proof}

The following lemma is a simplified rewriting of the condition introduced in Lemma~\ref{lem:foc-h-temp}.
\begin{lem}\label{lem:foc-h-modified}
  Let $\xsol \in \RR^N$, $I = \hsupp(\xsol)$ and $\mu = J_H(\xsol)$.
  Assume~\eqref{eq:ci} holds. 
  Let $U$ be any basis of $\Ker H_I^*$.
  There exists $z \in \Ker H_I$ such that
  \begin{equation*}
    v_I = z_I + \frac{1}{\lambda} \PhiIT^* \Gamma_I^\bot (y - \mu \PhiIT \UI) \in \SimplexI,
  \end{equation*}
  if, and only if, $\xsol$ is a solution of~\eqref{eq:lasso-h}.
  Moreover, if $v_I \in \relint \SimplexI$, then $\xsol$ is the unique solution of~\eqref{eq:lasso-h}.
\end{lem}
\begin{proof}
  Note that any vector $x \in \RR^N$ such that the condition~\eqref{eq:ci} holds, where $I$ is the $H$-support of $x$, is such that
  \begin{equation*}
    x = \mu H_I^{+,*} \UI + U \alpha \qwhereq \mu = J_H(x),
  \end{equation*}
  for some coefficients $\alpha$ and $U$ any basis of $\Ker H_I^*$.
  We obtain
  \begin{equation*}
    U \Phi^* (\Phi \xsol - y ) = \mu U \Phi^* \Phi H_I^{+,*} \UI - U\Phi^*y + U \Phi^* \Phi U \alpha = 0
  \end{equation*}
  Since~\eqref{eq:ci} holds, we have
  \begin{equation*}
    \alpha = (U \Phi^* \Phi U \alpha)^{-1} U \Phi^* \left( y - \mu \PhiIT \UI \right) .
  \end{equation*}
  Hence,
  \begin{equation*}
    \Phi U \alpha = \Gamma_I \left( y - \mu \PhiIT \UI \right) .
  \end{equation*}
  Now since, $\xsol = \mu H_I^{+,*} \UI + U \alpha$, one has
  \begin{equation*}
    \Phi \xsol = \mu \PhiIT \UI + \Gamma_I \left( y - \mu \PhiIT \UI \right) = \mu \Gamma_I^\bot \PhiIT \UI + \Gamma_I y .
  \end{equation*}
  Subtracting $y$ and multiplying by $\PhiIT^*$ both sides, and replacing in the expression of $v_I$ in Lemma~\ref{lem:foc-h-temp}, we get the desired result.
\end{proof}

\subsection{Proof of Theorem~\ref{thm:small-noise}}

  Let $I$ be the $H$-support of $x_0$.
  We consider the restriction of~\eqref{eq:lasso-h} to the $H$-support $I$.
  \begin{equation}\label{eq:lasso-h-restricted}\tag{$\Pp_\lambda(y)_I$}
    \xsol =
    \uargmax{\substack{x \in \RR^N \\ \hsupp(x) \subseteq I}}
    \dfrac{1}{2} \norm{y - \Phi x}_2^2
    + J_H(x) .
  \end{equation}
  Thanks to~\eqref{eq:ci}, the objective function is strongly convex on the set of signals of $H$-support I.
  Hence $\xsol$ is uniquely defined.
  The proof is divided in five parts: We give (\textbf{1.}) an implicit form of $\xsol$.
  We check (\textbf{2.}) that the $\H$-support of $\xsol$ is the same as the $\H$-support of $x_0$.
  We provide (\textbf{3.}) the value of $J_H(\xsol)$.
  Using Lemma~\ref{lem:foc-h-modified}, we prove (\textbf{4.}) that $\xsol$ is the unique minimizer of~\eqref{eq:lasso-h}.

  \textbf{1. Expression of $\xsol$.}
  One has $\xsol = \mu \H_I^{+,*} \UI + U \alpha$ where $\mu = J_H(\xsol)$.
  Hence,
  \begin{equation*}
    U^* \Phi^*(\Phi x - y) = \mu U^* \Phi^* \Phi \H_I^{+,*} \UI + (U^* \Phi^* \Phi U) \alpha - U^* \Phi^* y = 0 .
  \end{equation*}
  Thus,
  \begin{equation*}
    U \alpha = U M_I U^* \Phi^*( y - \mu \Phi H_I^{+,*} \UI) .    
  \end{equation*}
  Now, since $y = \Phi x_0 + w$, with $\hsupp(x_0) = I$, then
  \begin{align*}
    \xsol &= \mu H_I^{+,*} \UI + U M_I U^* \Phi^* (y - \mu \Phi H_I^{+,*} \UI) \\
          &= \mu H_I^{+,*} \UI + U M_I U^* \Phi^* ((\mu_0 - \mu) \Phi H_I^{+,*} \UI + w) + U \alpha_0 \\
          &= x_0 - (\mu_0 - \mu) H_I^{+,*} \UI + U M_I U^* \Phi^* ((\mu_0 - \mu) \Phi H_I^{+,*} \UI + w) ,
  \end{align*}
  where $\mu_0 = J_H(x_0)$.
  Hence, $\xsol$ is satisfying
  \begin{equation}\label{eq:sol-res-eq}
    \xsol = x_0 + (\mu_0 - \mu) [ U M_I U^* \Phi^* \Phi - \Id] H_I^{+,*} \UI + U M_I U^* \Phi^* w .
  \end{equation}

  \textbf{2. Checking that the $H$-support of $\xsol$ is $I$.}
  To ensure that the $H$-support of $\xsol$ is $I$ we have to impose that
  \begin{gather*}
    \forall i \in I,\quad \dotp{h_i}{\xsol} = J_H(\xsol) = \mu \\
    \forall j \in I^c,\quad \dotp{h_j}{\xsol} < J_H(\xsol) = \mu .
  \end{gather*}
  The components on $I$ of $\xsol$ are satisfying $H_I^* \xsol = \mu \UI$.
  Since $J_H$ is subadditive, we bound the components on $I^c$ by the triangular inequality on~\eqref{eq:sol-res-eq} to get
  \begin{align*}
    \max_{j \in I^c} \dotp{h_j}{\xsol} \leq
    &\max_{j \in I^c} \dotp{h_j}{x_0} \\
    &+ (\mu_0 - \mu) \normi{H_{I^c}^*[ U M_I U^* \Phi^* \Phi - \Id] H_I^{+,*} \UI} \\
    &+ \normi{H_{I^c}^* U M_I U^* \Phi^* w} .
  \end{align*}
  Denoting
  \begin{align*}
    C_1 &= \normi{H_{I^c}^*[ U M_I U^* \Phi^* \Phi - \Id] H_I^{+,*} \UI} ,\\
    C_2 &= \norm{H_{I^c}^* U M_I U^* \Phi^*}_{2,\infty} ,\\
    T &= \mu_0 - \max_{j \in {I^c}} \dotp{h_j}{x_0} ,
  \end{align*}
  we bound the correlations outside the $H$-support by
  \begin{equation*}
    \max_{j \in I^c} \dotp{h_j}{\xsol} \leq
    \mu_0 - T + (\mu_0 - \mu) C_1 + C_2 \norm{w} . 
  \end{equation*}
  There exists some constants $c_1, c_2$ satisfying $c_1 \norm{w} < c_2 T + \lambda$ such that
  \begin{equation}\label{eq:cond-2}
    0 \leq \mu_0 - T + (\mu_0 - \mu) C_1 + C_2 \norm{w} < \mu
  \end{equation}
  Under this condition, one has 
  \begin{equation*}
    \max_{j \in I^c} \dotp{h_j}{\xsol} < \mu ,
  \end{equation*}
  which proves that $\hsupp(\xsol) = I$.

  \textbf{3. Value of $\mu=J_H(\xsol)$.}
  Using Lemma~\ref{lem:foc-h-modified} with $H = U^* H$, since $\xsol$ is a solution of~\eqref{eq:lasso-h-restricted}, there exists $z_I \in \Ker H_I$ such that
  \begin{equation}\label{eq:inter-cond}
    v_I = z_I + \frac{1}{\lambda} \PhiIT^* \Gamma_I^\bot (y - \mu \PhiIT \UI) \in \SimplexI .
  \end{equation}
  We decompose $x_0$ as
  \begin{equation*}
    x_0 = \mu_0 H_I^{+,*} \UI + U \alpha_0 .
  \end{equation*}
  Since $y = \Phi x_0 + w$, we have
  \begin{equation*}
    \Gamma_I^\bot y = \Gamma_I^\bot ( \mu_0 \PhiIT \UI + \Phi U \alpha_0 + w) .
  \end{equation*}
  Now since
  \begin{equation*}
    \Gamma_I \Phi U \alpha_0 = \Phi U (U^* \Phi^* \Phi U)^{-1} U^* \Phi^* \Phi U \alpha_0 = \Phi U \alpha_0,
  \end{equation*}
  one obtains
  \begin{equation*}
    \Gamma_I^\bot y = \mu_0 \Gamma_I^\bot \PhiIT \UI + \Gamma_I^\bot w .
  \end{equation*}
  Thus, equation~\eqref{eq:inter-cond} equivalently reads
  \begin{equation*}
    v_I = z_I + \frac{1}{\lambda} \PhiIT^* \Gamma_I^\bot \left( (\mu_0 - \mu) \PhiIT \UI + w \right) .
  \end{equation*}
  In particular, $\dotp{v_I}{\UI} = \lambda$.
  Thus,
  \begin{equation*}
       \lambda 
    = \dotp{\lambda v_I}{\UI}
    = \dotp{\lambda \tilde z_I}{\UI} + \dotp{\PhiIT^* \Gamma_I^\bot((\mu_0 - \mu)\PhiIT \UI + w}{\UI} .
  \end{equation*}
  Since $\tilde z_I \in \Ker H_I$, one has $\dotp{z_I}{\UI} = 0$.
  \begin{align*}
       \lambda 
    &= \dotp{\PhiIT^* \Gamma_I^\bot((\mu_0 - \mu)\PhiIT \UI + w}{\UI} \\
    &= (\mu_0 - \mu) \norm{\PhiIT \UI}_{\Gamma_I^\bot}^2 + \dotp{\PhiIT \UI}{w}_{\Gamma_I^\bot} .
  \end{align*}
  Thus the value of $\mu$ is given by
  \begin{equation}\label{eq:sol-res-value}
    \mu = \mu_0 + \frac{\dotp{\PhiIT \UI}{w}_{\Gamma_I^\bot} - \lambda}{\norm{\PhiIT \UI}_{\Gamma_I^\bot}^2} > 0.
  \end{equation}

  \textbf{4. Checking conditions of Lemma~\ref{lem:foc-h-modified}.}
  Consider now the vector $\tilde v_I$ defined by
  \begin{equation*}
    \tilde v_I = \tilde z_I + \frac{1}{\lambda} \PhiIT^* \Gamma_I^\bot \left( (\mu_0 - \mu) \PhiIT \UI + w \right) ,
  \end{equation*}
  where
  \begin{equation*}
    \tilde z_I = 
    \frac{1}{\mu - \mu_0}
    \left(
    \uargmax{z_I \in \Ker H_I}
    \min_{i \in I} 
    (\PhiIT^* \Gamma_I^\bot \PhiIT \UI + z_I)_i
    \right)
  \end{equation*}
  Under condition~\eqref{eq:cond-2}, the $H$-support of $\xsol$ is $I$, hence we only have to check that $\tilde v_I$ is an element of $\relint \SimplexI$.
  Since $\dotp{\tilde z_I}{\UI} = 0$, one has
  \begin{align*}
    &\dotp{\tilde v_I}{\UI} \\
    =& \dotp{z_I + \frac{1}{\lambda} \PhiIT^* \Gamma_I^\bot \left( (\mu_0 - \mu) \PhiIT \UI + w \right)}{\UI}
       + \dotp{\tilde z_I - z_I}{\UI} \\
    =& \dotp{v_I}{\UI} + 0 = \lambda .
  \end{align*}
  Plugging back the expression~\eqref{eq:sol-res-value} of $(\mu_0 - \mu)$ in the definition of $\tilde v_I$, one has
  \begin{equation*}
    \tilde v_I = \tilde z_I + \frac{1}{\lambda}
          \left( \PhiIT^* \Gamma_I^\bot w
                  + \frac{\dotp{\PhiIT \UI}{w}_{\Gamma_I^\bot} - \lambda}{\norm{\PhiIT \UI}_{\Gamma_I^\bot}^2}
                    \PhiIT^* \Gamma_I^\bot \PhiIT \UI
           \right) .
  \end{equation*}
  For some constant $c_3$ such that $c_3 \norm{w} - \IC_H(I) \cdot \lambda > 0$, one has
  \begin{equation*}
    \forall i \in I,\quad v_i > 0 .
  \end{equation*}
  Combining this with the fact that $\dotp{\tilde v_I}{\UI} = \lambda$ proves that $\tilde v_I \in \relint \SimplexI$.
  According to Lemma~\ref{lem:foc-h-modified}, $\xsol$ is the unique minimizer of~\eqref{eq:lasso-h}.
\endIEEEproof

\subsection{Proof of Theorem~\ref{thm:noiseless}}

Taking $w = 0$ in Theorem~\ref{thm:small-noise}, we obtain immediately
\begin{lem}\label{lem:small-without}
  Let $x_0 \in \RR^N \setminus \ens{0}$ and $I$ its $H$-support such that~\eqref{eq:ci} holds.
  Let $y = \Phi x_0$.
  Suppose that $\PhiIT \UI \neq 0$ and $\IC_H(I) > 0$.
  Let $T = \umin{j \in I^c} J_H(x_0) - \dotp{x_0}{h_j} > 0$ and $\lambda < T \tilde c_I$.
  Then,
  \begin{equation*}
    \xsol = x_0 + \frac{\lambda}{\norm{\PhiIT \UI}_{\Gamma_I^\bot}^2} [UM_IU^*\Phi^*\Phi - \Id]H_I^{+,*} \UI ,
  \end{equation*}
  is the unique solution of~\eqref{eq:lasso-h}.
\end{lem}

The following lemma shows that under the same condition, $x_0$ is a solution of~\eqref{eq:bp-h}.
\begin{lem}\label{lem:noiseless-x0}
  Let $x_0 \in \RR^N \setminus \ens{0}$ and $I$ its $H$-support such that~\eqref{eq:ci} holds.
  Let $y = \Phi x_0$.
  Suppose that $\PhiIT \UI \neq 0$ and $\IC_H(I) > 0$.
  Then $x_0$ is a solution of~\eqref{eq:bp-h}.
\end{lem}
\begin{proof}
  According to Lemma~\ref{lem:small-without}, for every $0 < \lambda <T \tilde c_I$,
  \begin{equation*}
    \xsol_\lambda = x_0 + \frac{\lambda}{\norm{\PhiIT \UI}_{\Gamma_I^\bot}^2} [UM_IU^*\Phi^*\Phi - \Id]H_I^{+,*} \UI ,
  \end{equation*}
  is the unique solution of~\eqref{eq:lasso-h}.

  Let $\tilde x \neq x_0$ such that $\Phi \tilde x = y$.
  For every $0 < \lambda <T \tilde c_I$, since $\xsol_\lambda$ is the unique minimizer of~\eqref{eq:lasso-h}, one has
  \begin{equation*}
    \frac{1}{2} \norm{y - \Phi \xsol_\lambda}_2^2 + J_H(\xsol_\lambda)
    <
    \frac{1}{2} \norm{y - \Phi \tilde x}_2^2 + J_H(\tilde x) .
  \end{equation*}
  Using the fact that $\Phi \tilde x = y = \Phi x_0$, one has $J_H(\xsol_\lambda) < J_H(\tilde x)$.
  By continuity of the mapping $x \mapsto J_H(x)$, taking the limit for $\lambda \to 0$ in the previous inequality gives 
  \begin{equation*}
    J_H(x_0) \leq J_H(\tilde x).
  \end{equation*}
  It follows that $x_0$ is a solution of~\eqref{eq:bp-h}.
\end{proof}

We now prove Theorem~\ref{thm:noiseless}.
\begin{proof}[Proof of Theorem~\ref{thm:noiseless}]
  Lemma~\ref{lem:noiseless-x0} proves that $x_0$ is a solution of~\eqref{eq:bp-h}.
  We now prove that $x_0$ is in fact the unique solution.
  Let $\tilde z_I$ be the argument of the maximum in the definition of $\IC_H(I)$.
  We define
  \begin{equation*}
    \tilde v_I = \dfrac{1}{\norm{\PhiIT \UI}_{\Gamma_I^\bot}^2} 
                 \left( \tilde z_I + \PhiIT^* \Gamma_I^\bot \PhiIT \UI \right).
  \end{equation*}
  By definition of $\IC_H(I)$, for every $i \in I, \tilde v_I > 0$ and $\dotp{\tilde v_I}{\UI} = 1$.
  Thus, $H_I \tilde v_I \in \relint (\partial J_H(x_0))$.
  Moreover, since $\tilde z_I \in \Ker H_I$, one has
  \begin{equation*}
    H_I v_I = H_I H_I^{+,*} \Phi^* \Gamma_I^\bot \PhiIT \UI = \Phi^* \eta \qwhereq \eta = \Gamma_I^\bot \PhiIT \UI .
  \end{equation*}
  Thanks to Lemma~\ref{lem:uniqueness-h}, $x_0$ is the unique solution of~\eqref{eq:bp-h}.
\end{proof}


\bibliographystyle{IEEEtran}
\bibliography{IEEEabrv,vaiter_polyhedral_robustness}
\end{document}